\documentclass[conference]{IEEEtran}
\ifCLASSINFOpdf
\else
\fi
\usepackage{times} 
\usepackage{graphicx}                              % LaTeX 2e 
\usepackage{epsfig}
\usepackage{grffile}
\usepackage{algorithmic}
\usepackage{amssymb}
\usepackage{amsmath}
\usepackage[ruled]{algorithm}
\usepackage{slashbox}
\usepackage{stmaryrd}
\usepackage{flushend}

\newcommand{\argmax}{\operatornamewithlimits{argmax}}

\newtheorem{theorem}{Theorem}[section]
\newtheorem{lemma}{Lemma}[section]  

\newcommand{\reffig}[1]{Figure~\ref{#1}}

\newcommand{\refsec}[1]{Section~\ref{#1}}
\newcommand{\refthe}[1]{Theorem~\ref{#1}}
\newcommand{\reflem}[1]{Lemma~\ref{#1}}

\newcommand{\refeq} [1]{equation~(\ref{#1})}
\newcommand{\refeqwo}[1]{(\ref{#1})}

\newcommand{\bs}[1]{{\boldsymbol{#1}}}
\newcommand{\sym}{{{\psi}}}

\begin{document}
\title{Low Complexity Linear Programming Decoding of Nonbinary Linear Codes}

\author{Mayur Punekar and Mark F. Flanagan\\% <-this % stops a space
Claude Shannon Institute, \\
University College Dublin, Belfield, Dublin 4, Ireland. \\
{\{mayur.punekar, mark.flanagan\}@ieee.org}
%\thanks{This work was not supported by any organization}% <-this % stops a space
% \thanks{*Claude Shannon Institute, University College Dublin, Belfield, Dublin 4, Ireland, {\{mayur.punekar, mark.flanagn\}@ieee.org}}
%\thanks{H. Kwakernaak is with Faculty of Electrical Engineering, Mathematics and Computer Science,
        %University of Twente, 7500 AE Enschede, The Netherlands
        %{\tt\small h.kwakernaak@autsubmit.com}}%
%\thanks{P. Misra is with the Department of Electrical Engineering, Wright State University,
        %Dayton, OH 45435, USA
        %{\tt\small pmisra@cs.wright.edu}}%
}

%TODO: 
%DONE: Abstract, FFG in notation,llbracket notation, re-place dongles in FFGs, matlab figure, finish the intro, FFG figures for small code, Future work.

% \begin{document}

\maketitle
% \thispagestyle{empty}
% \pagestyle{empty}

%%%%%%%%%%%%%%%%%%%%%%%%%%%%%%%%%%%%%%%%%%%%%%%%%%%%%%%%%%%%%%%%%%%%%%%%%%%%%%%%
\begin{abstract}
Linear Programming (LP) decoding of Low-Density Parity-Check (LDPC) codes has attracted much attention in the research community in the past few years. The aim of LP decoding is to develop an algorithm which has error-correcting performance similar to that of the Sum-Product (SP) decoding algorithm, while at the same time it should be amenable to mathematical analysis. 
The LP decoding algorithm has been derived for both binary and nonbinary decoding frameworks.
However, the most important problem with LP decoding for both binary and nonbinary linear codes is that the complexity of standard LP solvers such as the simplex algorithm remain prohibitively large for codes of moderate to large block length. To address this problem, Vontobel \textit{et al.} proposed a low complexity LP decoding algorithm for binary linear codes which has complexity linear in the block length. In this paper, we extend the latter work and propose a low-complexity LP decoding algorithm for nonbinary linear codes. We use the LP formulation 
for the nonbinary codes 
as a basis and derive a pair of primal-dual LP formulations.
The dual LP is then used to develop the low-complexity LP decoding algorithm for nonbinary linear codes. 
The complexity of the proposed algorithm is linear in the block length and is limited mainly by the maximum check node degree. 
As a proof of concept, we also present a simulation result for a $[80,48]$ LDPC code defined over $\mathbb{Z}_4$ using quaternary phase-shift keying over the AWGN channel, and we show that the error-correcting performance of the proposed LP decoding algorithm is similar to that of the standard LP decoding using the simplex solver.
\end{abstract}
% for the nonbinary linear codes 
% using the methods proposed by Vontobel \textit{et al.} 
% Due to this the complexity of the partial computations is exponential in maximum check node degree.
% Due to this the complexity of the proposed algorithm is exponential in maximum check node degree. However, it is linear in block length.
%%%%%%%%%%%%%%%%%%%%%%%%%%%%%%%%%%%%%%%%%%%%%%%%%%%%%%%%%%%%%%%%%%%%%%%%%%%%%%%%
\section{INTRODUCTION}
Low-Density Parity-Check (LDPC) codes belong to the class of capacity achieving codes. They were introduced in the early 1960s by Gallager \cite{Ga_63} but attracted the attention of the research community only after they were rediscovered by MacKay \textit{et al.} in the late 1990s \cite{MaNe_96}. LDPC codes are generally decoded by the belief propagation algorithm (also known as the Sum-Product (SP) algorithm) which has time complexity linear in the block length. However, binary LDPC codes suffer from an \textit{error floor} effect in the high SNR region. Some progress has been made in the direction of finite length analysis of LDPC codes and concepts such as graph-cover pseudocodewords, trapping sets, stopping sets etc. were introduced and investigated to understand the behavior of the SP algorithm. Nevertheless, finite length analysis of LDPC codes under the SP algorithm is a difficult task and it is still difficult to predict error floor behavior for a particular code.
% There have been some attempts in that direction e.g. graph-cover pseudocodeword, trapping-set, stopping set etc are studied to . but it is still difficult to predict error-floor effect for a particular code. 

The main focus of research in the area of LDPC codes has been on \textit{binary} LDPC codes. However, it is desirable to use nonbinary LDPC codes in many applications where bandwidth efficient higher order (i.e. nonbinary) modulation schemes are used. Nonbinary LDPC codes are also considered for storage applications \cite{MaHa_09}. Nonbinary LDPC codes and the corresponding nonbinary SP algorithm were investigated by Davey and MacKay in \cite{DaMa_98} and since then many code construction methods and optimized nonbinary SP algorithms have been proposed. However, finite length analysis of nonbinary LDPC codes under the nonbinary SP algorithm is also difficult and very few attempts (e.g. \cite{AnKa_10}) have been made in this direction.

An alternative decoding algorithm for binary LDPC codes, known as Linear Programming (LP) decoding, was proposed by Feldman \textit{et al.} \cite{FeWa_05}. In LP decoding, the ML decoding problem is modeled as an Integer Programming (IP) problem which is then relaxed to obtain the corresponding LP problem. This LP problem is solved with the help of standard LP solvers such as simplex. Compared to SP decoding, LP decoding relies on the well-studied mathematical theory of LP. Hence, the LP decoding algorithm is better suited to mathematical analysis and it is possible to make statements about its complexity and convergence, as well as to place bounds on its error-correcting-performance etc. However, the worst-case time complexity of the simplex solver is known to be exponential in the number of variables, which limits the use of LP decoding algorithms to codes of small block length. To overcome the complexity problem,  in \cite{VoKo_06} Vontobel \textit{et al.} used techniques from LP and coding theory to derive a low-complexity LP decoding algorithm for approximate LP decoding. The complexity of this latter LP decoding algorithm is linear in the block length and similar to that of the SP algorithm. A similar algorithm for more general graphical models is proposed in \cite{GlJa_07}. An extension of low-complexity LP decoding algorithm of \cite{VoKo_06} was proposed and studied in \cite{Bu_09}.

In \cite{FlSk_09},
% Flanagan \textit{et al.} extended the work of Feldman \textit{et al.} \cite{FeWa_05} for nonbinary codes and introduced the LP decoding algorithm for nonbinary linear codes. 
the LP decoding algorithm for binary linear codes was extended to the case of nonbinary linear codes.
The nonbinary LP decoding algorithm of \cite{FlSk_09} also relies on the simplex LP solver and hence its complexity is prohibitively large for moderate and large block length codes. In this paper we extend the work of \cite{VoKo_06} and propose a low-complexity LP decoding algorithm for nonbinary linear codes. We use the LP formulation of nonbinary linear codes proposed in \cite{FlSk_09} to develop an equivalent primal LP formulation. Then using the the techniques introduced in \cite{Vo_02, VoLo_03},
the corresponding dual LP is derived 
% the primal LP is used to derive dual LP 
which in turn is used to develop an update equation for the low-complexity LP decoding algorithm. 
% for which the techniques introduced in \cite{VoKo_06} are used extensively. 
This paper follows the development of the low-complexity LP decoder for binary LDPC codes proposed in \cite{VoKo_06}. However, 
% there are two main contributions of this paper, 
%it differs from the \cite{VoKo_06} in two aspects, 
there are three main points in which it differs from the work in \cite{VoKo_06};
% there two main points which are we differ from \cite{VoKo_06} in two main 
% there are two main points that differ from the work in \cite{VoKo_06}, 
first, the derivation of the primal-dual LP formulations for nonbinary linear codes; second, the update equation for the low-complexity LP decoding of nonbinary codes; and third, the decision rule required to obtain an estimate of the symbols after the algorithm terminates.

The rest of the paper is structured as follows. We begin with some notation and background information in Section II. The primal LP is developed in Section III and the corresponding dual LP is given in Section IV. The notion of ``local function'' is given in Section V. Section VI presents the low-complexity LP decoding algorithm for nonbinary linear codes. Simulation results are presented and discussed in Section VII. Conclusions are given in Section VIII, along with future directions for this research.
%  however its error-correcting performance is to the LP decoding. This algorithm was further extended to 
% an algorithm which uses update equations from 
% LP is as a well studied branch of mathematics allows one to LP decoding algorithm % The main advantage of using LP methods lies in fact that, LP is a well studied branch of mathematics which allows us to make statement about e.g. complexity, convergence, error correcting performance etc.
% LP is a well studied branch of mathematics due to which the LP decoding algorithm are better suited to the mathematical analysis which allows us to make statement about e.g. complexity, convergence, error correcting performance etc.
% The main advantage of using LP methods for decoding lies in LP methods amenability to mathematical analysis which allows us to make to make statement about e.g. complexity, convergence, performance bounds etc for the LP decoding algorithm. Even with the advantage of 

%%%%%%%%%%%%%%%%%%%%%%%%%%%%%%%%%%%%%%%%%%%%%%%%%%%%%%%%%%%%%%%%%%%%%%%%%%%%%%%%
\section{Notation and Background}
Let $\Re$ be a finite ring with $q$ elements where $0$ denotes the additive identity, and let $\Re^{-} = \Re \setminus \{0\}$. Let $\mathcal{C}$ be a linear code of length $n$ over the ring $\Re$, defined by
\begin{eqnarray}
\mathcal{C} = \{ \bs{c} \in \Re^{n} : \bs{c} \mathcal{H}^{T} = \bs{0} \}
\end{eqnarray}
where $\mathcal{H}$ is a $m \times n$ parity-check matrix with entries from $\Re$. The rate of code $\mathcal{C}$ is given by $R(\mathcal{C}) = \log_q (|\mathcal{C}|) / n$. Hence, the code $\mathcal{C}$ can be referred as an $[n, \log_q(|\mathcal{C}|)]$ linear code over $\Re$. 

The set $\mathcal{J} = \{1,\ldots,m\}$ denotes row indices and the set $\mathcal{I} = \{1,\ldots,n\}$ denotes column indices of $\mathcal{H}$. We use $\mathcal{H}_j$ for the $j$-th row of $\mathcal{H}$ and $\mathcal{H}^i$ for the $i$-th column of $\mathcal{H}$. supp$(\bs{c})$ denotes the support of the vector $\bs{c}$. For each $j \in \mathcal{J}$, let $\mathcal{I}_j = \mbox{supp}(\mathcal{H}_j)$ and for each $i \in \mathcal{I}$, let $\mathcal{J}_i = \mbox{supp}(\mathcal{H}^i)$. Also let $d_j = |\mathcal{I}_j|$ and $d = \max_{j \in \mathcal{J}}\{d_j\}$. We define set $\mathcal{E} = \{(i,j) \in \mathcal{I} \times \mathcal{J} \; : \; j \in \mathcal{J}, i \in \mathcal{I}_j\} = \{(i,j) \in \mathcal{I} \times \mathcal{J} \; : \; i \in \mathcal{I}, j \in \mathcal{J}_i\}$. Moreover for each $j \in \mathcal{J}$ we define the local Single Parity Check (SPC) code %$\mathcal{C}_j = \{\bs{c} \in \Re^{n} | \sum_{i \in \mathcal{I}} \bs{c}_i \cdot \mathcal{H}_{j,i} = 0 \}$.
\begin{equation*}
\mathcal{C}_j = \left \{(b_i)_{i \in \mathcal{I}_j} : \sum_{i \in \mathcal{I}_j} b_i \cdot \mathcal{H}_{j,i} = 0 \right\} 
\end{equation*}
For each $i \in \mathcal{I}$, we denote by $\mathcal{A}_i \subseteq \Re^{|\{0\} \cup \mathcal{J}_i|}$ the repetition code of the appropriate length and indexing. In addition, we use the following notation introduced in 
\cite{VoKo_06}: for a statement $A$ we have $\llbracket$A$\rrbracket = 0$ if statement $A$ is true and $\llbracket A \rrbracket = +\infty$ otherwise. Here $\llbracket A \rrbracket = -\log[A]$ and $[A]$ is Iverson's convection i.e. we have $[A] = 1$ if $A$ is true and $[A] = 0$ otherwise. Please note that where $A$ indicates the value of a variable, Iverson's convention can also be interpreted as the Kronecker delta function.
% 
% 
% 
% We might require following for journal version
% Here local code $\mathcal{A}_i \subseteq \Re^{|\{0\} \cup \mathcal{J}_i|}$ is the repetition code of length $|\mathcal{J}_i| + 1$ and $\mathcal{C}_j \subseteq \Re^{|\mathcal{I}_j|}$ is the single parity-check code $\mathcal{C}_j$. 
% (b_i)_{i \in \mathcal{I}_j} : \sum_{i \in \mathcal{I}_j} b_i \cdot \mathcal{H}_{j,i} = 0\right \}$.
% For any $\bs{c} \in \Re^n$, the parity check $j \in \mathcal{J}$ is satisfied iff
% \begin{equation}
% \bs{c}\mathcal{H}_j^{T} = \sum_{i \in \mathcal{I}_j} c_i \cdot \mathcal{H}_{j,i} = 0.
% \end{equation}
% For $j \in \mathcal{J}$, the single parity-check code $\mathcal{C}_j$ over $\Re$ is defined by
% \begin{equation*}
% \mathcal{C}_j = \left \{ (b_i)_{i \in \mathcal{I}_j} : \sum_{i \in \mathcal{I}_j} b_i \cdot \mathcal{H}_{j,i} = 0\right \}. 
% \end{equation*}
% The projection mapping for parity-check $j \in \mathcal{J}$ is defined by
% \begin{equation*}
% \bs{x}_j(\bs{c}) = (c_i)_{i \in \mathcal{I}_j}.\\
% \end{equation*}
% Then, for any given $\bs{c} \in \Re^{n}$, we may say that parity-check $j \in \mathcal{J}$ is satisfied by $\bs{c}$ iff
% \begin{equation}
% \bs{x}_j(\bs{c}) \in \mathcal{C}_j 
% \end{equation}
% From this we can say that $\bs{c} \in \mathcal{C}$ iff all parity checks $j \in \mathcal{J}$ are satisfied by $\bs{c}$. In this case, we say that $\bs{c}$ is a \textit{codeword} of $\mathcal{C}$.
% 
% 
% 
We define the following mapping as in \cite{FlSk_09},
\begin{equation*}
\xi : \Re \rightarrow \{0,1\} ^{q-1} \subset \mathbb{R}^{q-1}
\end{equation*}
by
\begin{equation*}
\xi(\alpha) = \bs{x} = (x^{(\rho)})_{\rho \in \Re^{-}}
\end{equation*}
such that, for each $\rho \in \Re^{-}$
\begin{eqnarray*}
x^{(\rho)} = \left\{ 
\begin{array}{l l}
  1, & \; \text{if} \; \rho = \alpha \\
  0, & \; \text{otherwise}\\
\end{array} \right.
\end{eqnarray*}
Building on this we define
\[
\Xi : \underset{t \in \mathbb{Z}^{+}}{\cup} \Re^{t} \rightarrow \underset{t \in \mathbb{Z}^{+}}{\cup} \{0,1\}^{(q-1)t} \subset \underset{t \in \mathbb{Z}^{+}}{\cup} \mathbb{R}^{(q-1)t} \; ,
\]
according to
\[
\Xi (\bs{c}) = (\xi(c_1),\dots,\xi(c_t))\nonumber, \quad \forall \bs{c} \in \Re^{t}, t \in \mathbb{Z}^{+} \; .
\]

For $\kappa \in \mathbb{R}, \kappa > 0$, we define the function \[\sym (x) = e^{\kappa x},\] and its inverse \[\sym^{-1} (x) = \frac{1}{\kappa} \log (x).\] 

For vectors $\bs{f} \in \mathbb{R}^{(q-1)n}$ we use the notation
\begin{equation*}
\bs{f} = (\bs{f}_1 \;| \; \bs{f}_2 \; | \; \cdots \; | \; \bs{f}_n) \quad \text{where} \quad
% \end{equation*}
% 
% \begin{equation*}
\forall i \in \mathcal{I}, \bs{f}_i = (f_{i}^{(\alpha)})_{\alpha \in \Re^{-}}
\end{equation*}

We also define the inverse of $\Xi$ as
\begin{eqnarray*}
\Xi^{-1}(\bs{f}) = (\xi^{-1}(\bs{f}_1), \xi^{-1}(\bs{f}_2), \cdots, \xi^{-1}(\bs{f}_n)).
\end{eqnarray*}
Note that the inverse of $\Xi$ is well defined for any $\bs{f} \in \mathbb{R}^{(q-1)n}$ where each component $\bs{f}_i$, $i  \in \mathcal{I}$, has entries from $\{ 0,1 \}$ with sum at most $1$. 
% We assume that the codeword $\bar{\bs{c}} = (\bar{c}_1, \bar{c}_2, \cdots, \bar{c}_n) \in \mathcal{C}$ has been transmitted over q-ary input memoryless channel and 
We assume transmission over a $q$-ary input memoryless channel and also assume a corrupted codeword $\bs{y} =(y_1, y_2, \cdots, y_n)\in \Sigma^{n}$ has been received. Here, the channel output symbols are denoted by $\Sigma$. Based on this, we define a function $\bs{\lambda} : \Sigma \rightarrow (\mathbb{R} \cup \{\pm\infty\})^{q-1}$ by
%\begin{eqnarray*}
%\bs{\lambda}  &=& (\lambda^{(\alpha)})_{\alpha \in \Re^{-}}  \quad
%\end{eqnarray*}
%$\text{where, for each } \; y \in \Sigma, \alpha \in \Re^{-}, $
%\begin{eqnarray*}
%\lambda^{(\alpha)}(y) &=& \log \left(\frac{p(y|0)}{p(y|\alpha)}\right) \; .
%\end{eqnarray*}
\[
\bs{\lambda} = (\lambda^{(\alpha)})_{\alpha \in \Re^{-}} 
\]
where, for each $y \in \Sigma$, $\alpha \in \Re^{-}$,
\[
\lambda^{(\alpha)}(y) = \log \left(\frac{p(y|0)}{p(y|\alpha)}\right) \; .
\]

Here $p(y|c)$ denotes the channel output probability (density) conditioned on the channel input. Based on this, we also define 
\begin{equation*}
\bs{\Lambda}(\bs{y}) = (\bs{\lambda}(y_1)\;|\;\bs{\lambda}(y_2)\;| \cdots |\;\bs{\lambda}(y_n)) \; . 
\end{equation*}
We will use Forney-style factor graphs (FFGs), also known as Normal graphs \cite{Fo_01} to represent the linear programs introduced in this paper. An FFG is a diagram that represents the factorization of a function of several variables. For more information on FFGs the reader is referred to \cite{Fo_01,Vo_02,LO_04}.

%%%%%%%%%%%%%%%%%%%%%%%%%%%%%%%%%%%%%%%%%%%%%%%%%%%%%%%%%%%%%%%%%%%%%%%%%%%%%%%%
\section{The Primal Linear Program}
% The Maximum Likelihood (ML) decoding is defined by,
% \begin{eqnarray}
% \hat{\bs{c}} = \argmax_{\bs{c} \in \mathcal{C}} p(\bs{y}|\bs{c}) \nonumber
% \end{eqnarray}
In \cite{FlSk_09} the authors presented the following linear program to decode nonbinary linear codes:

\medskip
\textbf{NBLPD} (Polytope $\mathcal{Q}_f$):
\begin{align*}
% \mbox{\textbf{NBLPD}} \quad & \text{(Polytope $Q_f$):}  & &\\
\text{min.} & \quad \bs{\Lambda}(\bs{y})\bs{f}^{T} & &\\  
\mbox{Subj. to }& &\\
f_{i}^{(\alpha)} &= \sum_{\underset{b_i = \alpha}{\bs{b} \in \mathcal{C}_j}} w_{j,\bs{b}} & \forall &j \in \mathcal{J}, \; \forall i \in \mathcal{I}_j, \; \forall \alpha \in \Re^{-}\\
w_{j,\bs{b}} &\ge 0 &\forall &j \in \mathcal{J}, \forall \bs{b} \in \mathcal{C}_j,\\
\sum_{\bs{b} \in \mathcal{C}_j} w_{j,\bs{b}} &= 1 &\forall &j \in \mathcal{J}.
\end{align*}
% \medskip
We denote the polytope represented by the constraints of NBLPD as $\mathcal{Q}_{f}$. Two alternative polytope representations are also given in \cite{FlSk_09}, which are both equivalent to NBLPD. It also possible to reformulate the constraints of NBLPD with additional auxiliary variables. % Along with the polytopes, an LP decoding algorithm based on $\mathcal{Q}_{f}$ for nonbinary linear codes is proposed in \cite{FlSk_09}. This algorithm uses LP solver such as Simplex method, which has worst case time complexity exponential in number of variables $n$. Alternatively, one can use solvers based on interior-point methods, which has time complexity polynomial in $n$. 
However, to develop a low-complexity LP decoding algorithm for NBLPD, we use the approach of \cite{VoKo_06} and reformulate NBLPD so that the new LP formulation can be directly represented by an FFG:

\medskip
\textbf{PNBLPD} (Polytope $\mathcal{Q}_p$):
\begin{align*}
\text{min.}& \quad\bs{\Lambda}(\bs{y}) \bs{f}^T & &\\
\mbox{Subj. to} & & & \nonumber \\
\bs{{f}}_i &= \bs{u}_{i,0} &(&i \in \mathcal{I}),  \nonumber \\
\bs{u}_{i,j} &= \bs{v}_{j,i} &(&(i,j) \in \mathcal{E}), \nonumber\\
\sum_{\bs{a} \in \mathcal{A}_i} \gamma_{i,\bs{a}} \; \Xi(\bs{a}) &= \bs{u}_i &(&i \in \mathcal{I}), \nonumber \\
\sum_{\bs{b} \in \mathcal{C}_j} \beta_{j,\bs{b}} \; \Xi(\bs{b}) &= \bs{v}_j &(&j \in \mathcal{J}), \nonumber \\
\gamma_{i,\bs{a}} &\geq 0 &(&i \in \mathcal{I}, \bs{a} \in \mathcal{A}_i), \nonumber \\
\beta_{j,\bs{b}} &\geq 0 &(&j \in \mathcal{J}, \bs{b} \in \mathcal{C}_j), \nonumber\\
\sum_{\bs{a} \in \mathcal{A}_i} \gamma_{i,\bs{a}} &= 1 &(&i \in \mathcal{I}), \nonumber \\
\sum_{\bs{b} \in \mathcal{C}_j} \beta_{j,\bs{b}} &= 1 &(&j \in \mathcal{J}). \nonumber
\end{align*}
Here $\bs{u}_{i,j} = (u_{i,j}^{(\alpha)})_{\alpha \in \Re^{-}}$ and $\bs{v}_{j,i} = (v_{j,i}^{(\alpha)})_{\alpha \in \Re^{-}}$ for all $i \in \mathcal{I}$, $j \in \mathcal{J}_i \cup \{ 0 \}$; also for $i \in \mathcal{I}$, $\bs{u}_i = (\bs{u}_{i,j})_{j \in \mathcal{J}_i \cup \{ 0 \}}$ and for $j \in \mathcal{J}$, $\bs{v}_j = (\bs{v}_{j,i})_{i \in \mathcal{I}_j}$. 
We denote the polytope represented by the constraints of PNBLPD by $\mathcal{Q}_{p}$. It is important to note that along with the convex hull of the single parity-check code, PNBLPD also explicitly models the convex hull of the repetition code. % related to the variable node. 
% PNBLPD has an advantage that the parts of it can be directly represented by the FFG (ref figure). 
% The FFG of ref fig. shows the part of PNBLPD.
% As we will see later on, this one-to-one relationship of PNBLPD with FFG is useful to derive appropriate dual linear program. 
% It might not be immediately clear to the reader but
The constraints of NBLPD and PNBLPD appear to be quite different due to the different notations. However, the projection of each polytope onto the variables denoted by $\bs{f}$ is the same in both cases, and therefore the LPs are equivalent from the point of view of decoding. The proof of their equivalence is given in \refthe{the:th1}.
% 
% 
% Before deriving the dual linear program we prove the equivalence of polytope $\mathcal{Q}_{f}$ and $\mathcal{Q}_{p}$.  
% are equivalent i.e. for every pair of $ (\bs{f}, \bs{w}) \in \mathcal{Q}_{f}$ we can find corresponding pair $ (\bs{f}, \bs{u}, \bs{v}) \in \mathcal{Q}_{p}$. 
%the projections of $\mathcal{Q}_{f}$ and $\mathcal{Q}_{p}$ onto the $(f_i)_{i \in \mathcal{I}}$ variables are the same. 
% The set $\bar{\mathcal{Q}}_{f} =  \{\bs{f} : \exists \; \bs{w} \; \mbox{s.t.} (\bs{f}, \bs{w}) \in \mathcal{Q}_{f}\}$ is equal to the set $\bar{\mathcal{Q}}_{p} =  \{\bs{f} : \exists \; \bs{u, v} \; \mbox{s.t.} (\bs{f}, \bs{u}, \bs{v}) \in \mathcal{Q}_{p}\}$.
% 
% 
\begin{theorem} \label{the:th1}
% Polytopes $\mathcal{Q}_{f}$ and $\mathcal{Q}_{p}$ are equivalent i.e. for every point $(\bs{f}, \bs{\gamma}, \bs{\beta}) \in \mathcal{Q}_{p}$ we can always find a corresponding point $(\bs{f}, \bs{w}) \in \mathcal{Q}_{f}$ and vice versa. 
Polytopes $\mathcal{Q}_{f}$ and $\mathcal{Q}_{p}$ are equivalent from an LP decoding perspective, i.e. for every $(\bs{f}, \bs{\gamma}, \bs{\beta}) \in \mathcal{Q}_{p}$ there exists $\bs{w}$ such that $(\bs{f}, \bs{w}) \in \mathcal{Q}_{f}$ and conversely, for every $(\bs{f}, \bs{w}) \in \mathcal{Q}_{f}$ there exist $\bs{\gamma}, \bs{\beta}$ such that $(\bs{f}, \bs{\gamma}, \bs{\beta}) \in \mathcal{Q}_{p}$.
\end{theorem}
\begin{proof}
% We observe that $\mathcal{C}_j$ and $\mathcal{C}_j$ are same. Now 
Suppose we have $(\bs{f}, \bs{\gamma}, \bs{\beta}) \in \mathcal{Q}_{p}$ and we define 
\begin{equation}
w_{j,\bs{b}} = \beta_{j,\bs{b}}, \quad \forall \bs{b} \in \mathcal{C}_j, \forall j \in \mathcal{J} \; .
\label{w_beta_eq} 
\end{equation}
% 
% 
% For $(\bs{f}, \bs{w}) \in \mathcal{Q}_{f}$, $(\bs{f}, \bs{w})$ must fulfill following constraints of NBLPD,
% \begin{eqnarray*}
% \text{(1)} && w_{j,\bs{b}} \geq 0, \quad \quad \quad \quad \quad \forall \bs{b} \in \mathcal{C}_j, \forall j \in \mathcal{J} \\
% \text{(2)} && \sum_{\bs{b} \in \mathcal{C}_j} w_{j,\bs{b}} = 1, \quad \quad \quad \forall j \in \mathcal{J} \\
% \text{(3)} && \sum_{\bs{b} \in \mathcal{C}_j, \bs{b}_{i} = \alpha} w_{j,\bs{b}} = f_i^{(\alpha)}, \forall \alpha \in \Re^{-}, \forall j \in \mathcal{J}, \forall i \in \mathcal{J}_i
% \end{eqnarray*}
% 
% 
The final two constraints of NBLPD are obviously fulfilled.
% 
% 
% \begin{eqnarray}
% \beta_{j,\bs{b}} &\geq& 0, \quad \forall \bs{b} \in \mathcal{C}_j, j \in \mathcal{J} \label{beta} \\
% \sum_{\bs{b} \in \mathcal{C}_j} \beta_{j,\bs{b}} &=& 1, \quad  \forall j \in \mathcal{J} \label{sum_beta}
% \end{eqnarray}
% Using \refeq{w_beta_eq} in \refeqwo{beta} \& \refeqwo{sum_beta} we get following,
% \begin{eqnarray*}
% % w_{j,\bs{b}} &\geq& 0, \quad \forall \bs{b} \in \mathcal{C}_j, j \in \mathcal{J}\\
% % \sum_{\bs{b} \in \mathcal{C}_j} w_{j,\bs{b}} &=& 1, \quad  \forall j \in \mathcal{J} 
% w_{j,\bs{b}} \geq 0, \quad \forall \bs{b} \in \mathcal{C}_j, j \in \mathcal{J} \quad \text{and} \quad \sum_{\bs{b} \in \mathcal{C}_j} w_{j,\bs{b}} = 1, \quad  \forall j \in \mathcal{J} 
% \end{eqnarray*}
% which fulfills the first two requirements. 
% 
% 
From PNBLPD, the following holds for $(\bs{f}, \bs{\gamma}, \bs{\beta}) \in \mathcal{Q}_{p}$:
\begin{align}
\bs{v}_j &= \sum_{\bs{b} \in \mathcal{C}_j} \beta_{j,\bs{b}} \; \Xi(\bs{b}), &\forall &j \in \mathcal{J}, \nonumber \\
\Rightarrow \bs{v}_{j,i} &= \sum_{\bs{b} \in \mathcal{C}_j} \beta_{j,\bs{b}} \; \xi(\bs{b}_i), &\forall &i \in \mathcal{J}_i, \forall j \in \mathcal{J},\nonumber %\\
\end{align}
This yields
\begin{align}
% \Rightarrow {v}_{j,i}^{(\alpha)} &= \sum_{\underset{\bs{b}_{i} = \alpha}{\bs{b} \in \mathcal{C}_j}} \beta_{j,\bs{b}}, &\forall &\alpha \in \Re^{-}, \forall i \in \mathcal{I}_j, \forall j \in \mathcal{J}, \nonumber \\
% \Rightarrow 
{v}_{j,i}^{(\alpha)} &= u_{i,j}^{(\alpha)} = \sum_{\underset{\bs{b}_{i} = \alpha}{\bs{b} \in \mathcal{C}_j}} \beta_{j,\bs{b}}, &\forall &\alpha \in \Re^{-}, \forall i \in \mathcal{I}_j, \forall j \in \mathcal{J}. \label{v_j_i}
\end{align}

% As we have repetition code for each variable node $i \in \mathcal{I}$, 
%All symbols of the codeword $\bs{a}$ of repetition code $\mathcal{A}_i, i \in \mathcal{I}$ must be same. Hence, %all the elements in vector $\bs{u}_i$ must be same too, i.e
From the third constraint of PNBLPD, and noting that $\mathcal{A}_i$ is a repetition code for each $i \in \mathcal{I}$, we have 
$u_{i,0}^{(\alpha)} = u_{i,j}^{(\alpha)}, \forall \alpha \in \Re^{-}, j \in \mathcal{J}_i$. With this and equation \eqref{v_j_i} we obtain the following,
% \begin{eqnarray}
% \Rightarrow u_{i,j}^{(\alpha)} &=& u_{i,0}^{(\alpha)} = \sum_{\bs{b} \in \mathcal{C}_j, \bs{b}_i = \alpha} \beta_{j,\bs{b}} \;  \nonumber\\
% % \Rightarrow x_{i}^{(\alpha)} &=& \sum_{\bs{b} \in \mathcal{C}_j, \bs{b}_i = \alpha} \beta_{j,\bs{b}} \nonumber\\
% \Rightarrow f_{i}^{(\alpha)} &=& \sum_{\bs{b} \in \mathcal{C}_j, \bs{b}_i = \alpha} w_{j,\bs{b}}, \quad \forall \alpha \in \Re^{-}, \forall i \in \mathcal{I}_j, \forall j \in \mathcal{J}. \label{point_3}
% \end{eqnarray}
\begin{align}
\Rightarrow u_{i,j}^{(\alpha)} &= u_{i,0}^{(\alpha)} = \sum_{\bs{b} \in \mathcal{C}_j, \bs{b}_i = \alpha} \beta_{j,\bs{b}} \nonumber \\
\Rightarrow f_{i}^{(\alpha)} &= \sum_{\underset{\bs{b}_{i} = \alpha}{\bs{b} \in \mathcal{C}_j}} w_{j,\bs{b}}, \quad \forall \alpha \in \Re^{-}, \forall i \in \mathcal{I}_j, \forall j \in \mathcal{J}. \label{point_3}
\end{align}
This proves the first constraint of NBLPD and hence $(\bs{f}, \bs{w}) \in \mathcal{Q}_{f}$.

The converse part of the theorem statement can be proved in a similar manner; the details are omitted. Since for a given vector $\bs{f}$, the objective functions in both formulations always have the same value, the decoding performance of NBLPD and PNBLPD are identical.
% As a side note, we observe that for a given vector $\bs{f}$, the objective functions in both formulations always have the same value. %since $u_{i,0}^{(\alpha)} = f_{i}^{(\alpha)}, \forall \alpha \in \Re^{-}, \forall i \in \mathcal{I}$, 
\end{proof}

\begin{figure*}
\begin{minipage}[b]{0.45\linewidth}
\centering
\includegraphics[width=1.0\columnwidth, keepaspectratio]{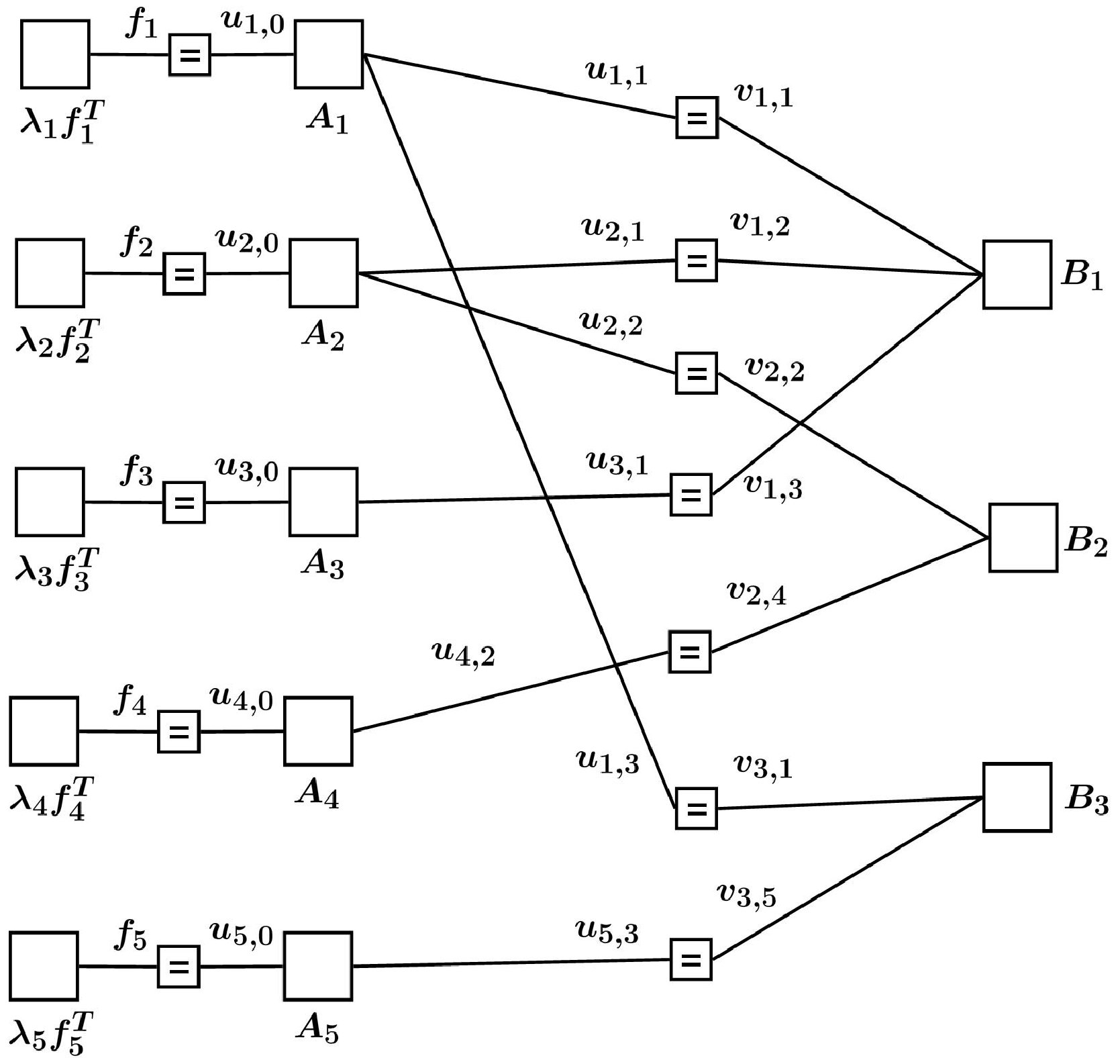}
\caption{FFG which represents the augmented cost function of \refeq{eq:primal} for the example $(5,2)$ binary code.} 
\label{fig:primal_LP_FFG}
\end{minipage}
\hspace{1cm}
\begin{minipage}[b]{0.45\linewidth}
\centering 
\includegraphics[width=1.0\columnwidth, keepaspectratio]{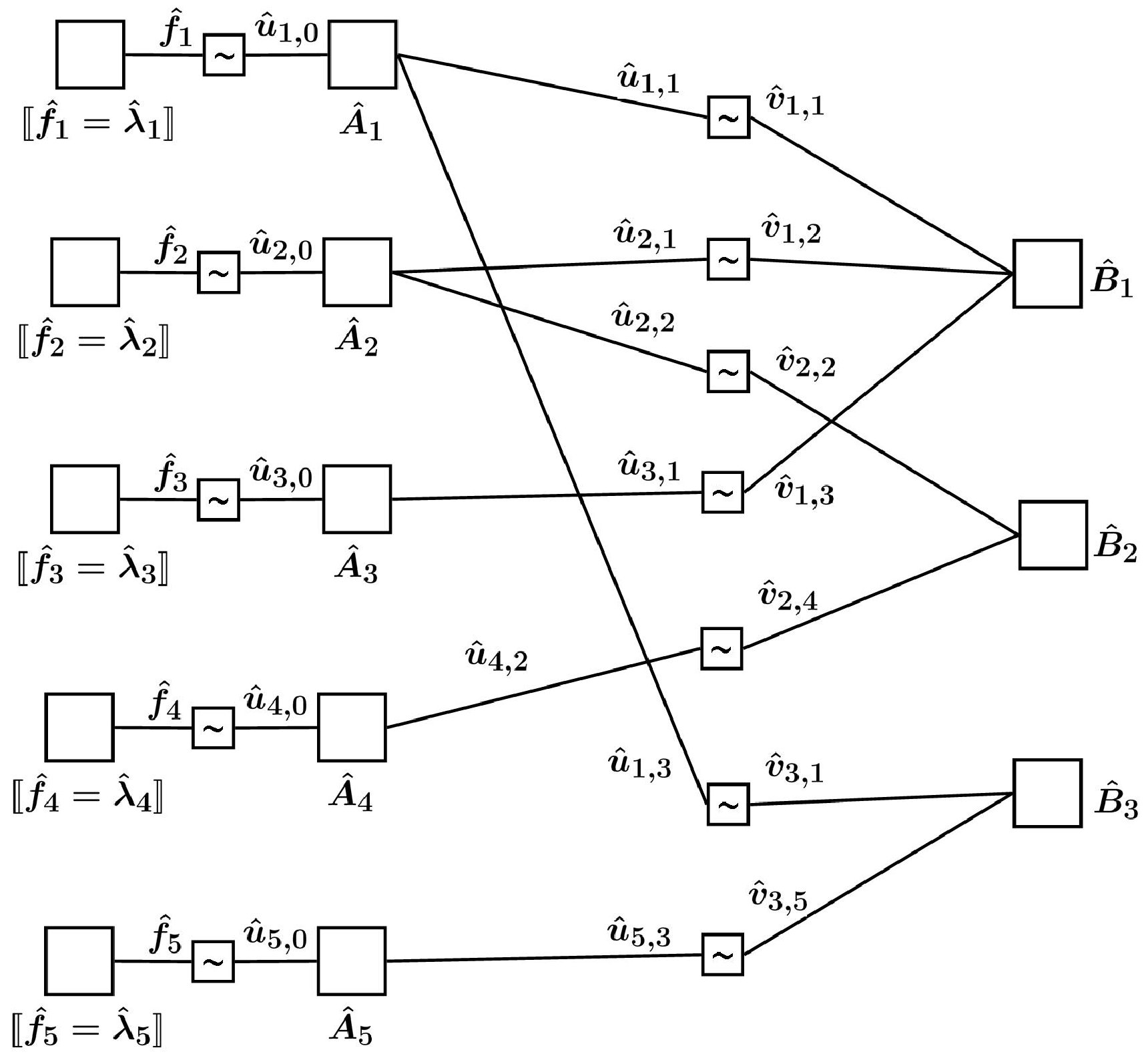}
\caption{FFG which represents the augmented cost function of \refeq{eq:dual} for the example $(5,2)$ binary code.} 
\label{fig:dual_LP_FFG}
\end{minipage}
\end{figure*}
\medskip

Before deriving the dual linear program, we reformulate the PNBLPD so that this LP can be represented by an FFG. For this purpose, constraints of the PNBLPD are expressed as additive cost terms (also known as \textit{penalty terms}). The rule for assigning cost to a configuration of variables is: if a given configuration satisfies the LP constraints then cost $0$ is assigned to this configuration, otherwise $+\infty$ is assigned. The PNBLPD is then equivalent to the unconstrained minimization of the following augmented cost function,
% \begin{eqnarray}
% \sum_{i \in \mathcal{I}} \bs{\lambda}_i \bs{f}_{i}^{T} + \sum_{i \in \mathcal{I}} \llbracket \bs{f}_i = \bs{u}_{i,0} \rrbracket + \sum_{(i,j) \in \mathcal{E}} \llbracket \bs{u}_{i,j} = \bs{v}_{j,i} \rrbracket + \sum_{i \in \mathcal{I}} A_i(\bs{u}_i) + \sum_{j \in \mathcal{J}} B_i(\bs{v}_i) \label{eq:primal}
% \end{eqnarray}
\begin{align}
\sum_{i \in \mathcal{I}} \bs{\lambda}_i \bs{f}_{i}^{T} &+ \sum_{i \in \mathcal{I}} \llbracket \bs{f}_i = \bs{u}_{i,0} \rrbracket + \sum_{(i,j) \in \mathcal{E}} \llbracket \bs{u}_{i,j} = \bs{v}_{j,i} \rrbracket \nonumber \\
&+ \sum_{i \in \mathcal{I}} A_i(\bs{u}_i) + \sum_{j \in \mathcal{J}} B_j(\bs{v}_j) \label{eq:primal}
\end{align}
where $\forall i \in \mathcal{I}$ and $\forall j \in \mathcal{J}$ we use
% 
% 
% \begin{eqnarray*}
% A_i(\bs{u}_i) &\triangleq& \left \llbracket \sum_{\bs{a} \in \mathcal{A}_i} \gamma_{i, \bs{a}} \; \Xi(\bs{a}) = \bs{u}_i \right \rrbracket + \sum_{\bs{a} \in \mathcal{A}_i} \llbracket \gamma_{i, \bs{a}} \ge 0 \rrbracket + \left \llbracket \sum_{\bs{a} \in \mathcal{A}_i} \gamma_{i, \bs{a}} = 1 \right \rrbracket , \label{eq:a_u_i} \\ \nonumber \\
% B_j(\bs{v}_j) &\triangleq& \left \llbracket \sum_{\bs{b} \in \mathcal{C}_j} \beta_{j, \bs{b}} \; \Xi(\bs{b}) = \bs{v}_j \right \rrbracket + \sum_{\bs{b} \in \mathcal{C}_j} \llbracket \beta_{j, \bs{b}} \ge 0 \rrbracket + \left \llbracket \sum_{\bs{b} \in \mathcal{C}_j} \beta_{j, \bs{b}} = 1 \right \rrbracket \; . 
% \end{eqnarray*}
% 
% 
\begin{align*}
A_i(\bs{u}_i) &\triangleq \left \llbracket \sum_{\bs{a} \in \mathcal{A}_i} \gamma_{i, \bs{a}} \; \Xi(\bs{a}) = \bs{u}_i \right \rrbracket + \sum_{\bs{a} \in \mathcal{A}_i} \llbracket \gamma_{i, \bs{a}} \ge 0 \rrbracket \\
&+ \left \llbracket \sum_{\bs{a} \in \mathcal{A}_i} \gamma_{i, \bs{a}} = 1 \right \rrbracket , \\
B_j(\bs{v}_j) &\triangleq \left \llbracket \sum_{\bs{b} \in \mathcal{C}_j} \beta_{j, \bs{b}} \; \Xi(\bs{b}) = \bs{v}_j \right \rrbracket + \sum_{\bs{b} \in \mathcal{C}_j} \llbracket \beta_{j, \bs{b}} \ge 0 \rrbracket \\
&+ \left \llbracket \sum_{\bs{b} \in \mathcal{C}_j} \beta_{j, \bs{b}} = 1 \right \rrbracket.
\end{align*}
% % % With this, PNBLPD is converted to a unconstrained optimization problem which allows us to represent constraints of PNBLPD in FFG.
For ease of illustration we consider a $(5,2)$ binary code with parity-check matrix 
\begin{eqnarray*}
\mathcal{H} = \left[ 
\begin{array}{l l l l l}
1 & 1 & 1 & 0 & 0 \\
0 & 1 & 0 & 1 & 0 \\
1 & 0 & 0 & 0 & 1 \\
\end{array}
\right] \; .
\end{eqnarray*}
The augmented cost function for this code is represented by the FFG of \reffig{fig:primal_LP_FFG}.

%%%%%%%%%%%%%%%%%%%%%%%%%%%%%%%%%%%%%%%%%%%%%%%%%%%%%%%%%%%%%%%%%%%%%%%%%%%%%%%%
\section{Dual Linear Program}
The dual linear program of PNBLPD can be derived from the augmented cost function of \refeq{eq:primal}. 
%In order to derive dual of \refeq{eq:primal}, 
First we derive the dual of $A_i(\bs{u}_i) \text{ and } B_j(\bs{v}_j)$. For simplicity of exposition, we assume $\mathcal{A}_i = \{\bs{p}, \bs{q}\} = \{(p_0, p_1, p_2), (q_0, q_1, q_2)\}$. The (primal) FFG of $A_i(\bs{u}_i)$ is shown in \reffig{fig:primal_FFG} and its dual is shown in \reffig{fig:dual_FFG}. The dual FFG is derived with the help of techniques introduced in \cite{Vo_02} and \cite{VoLo_03}. 
\begin{figure*}
\begin{minipage}[b]{0.45\linewidth}
\centering
\includegraphics[width=1.0\columnwidth, keepaspectratio]{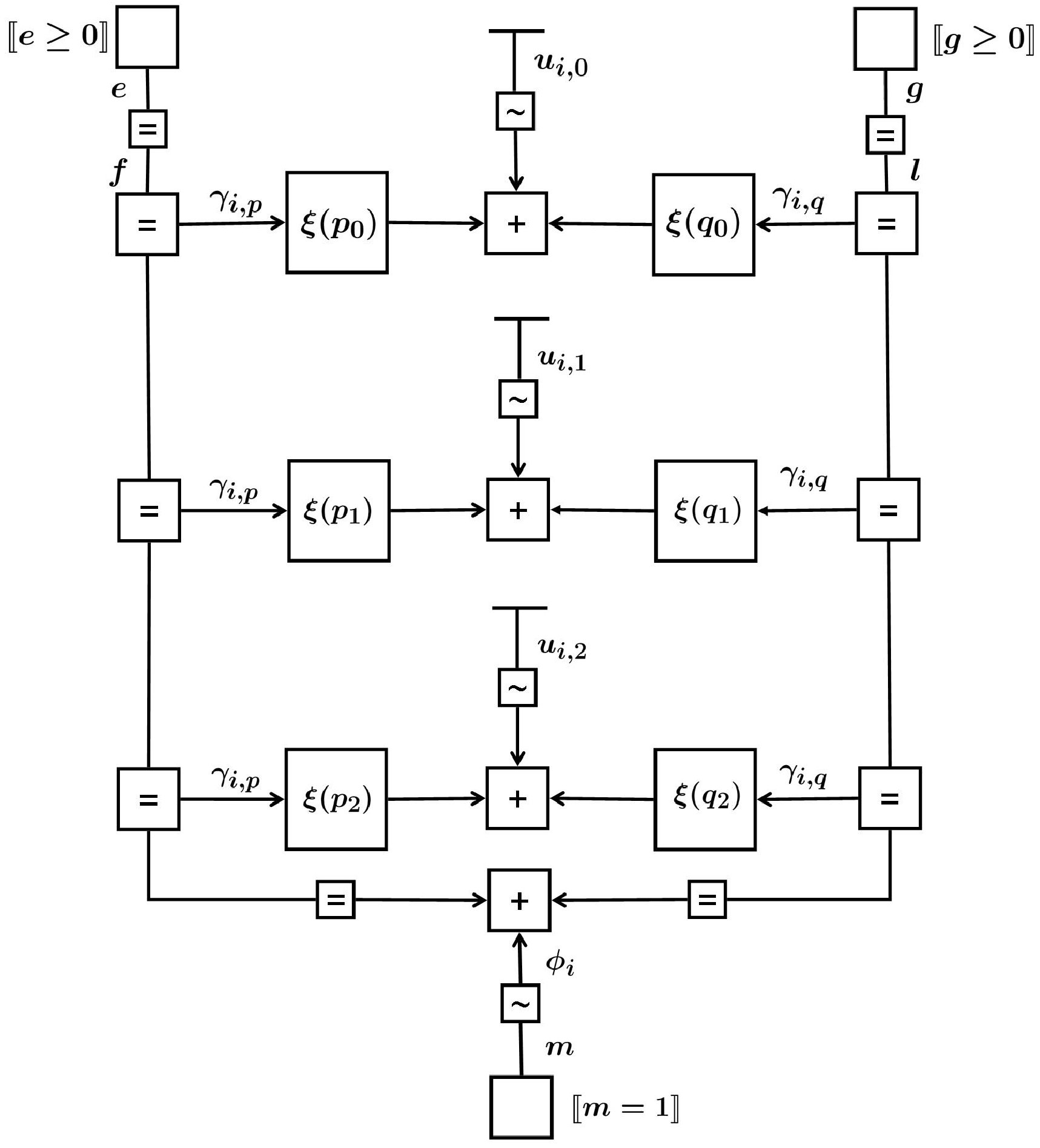}
\caption{FFG for the function $A_i(\bs{u}_i)$. This forms a subgraph of the overall FFG of \reffig{fig:primal_LP_FFG}.} 
\label{fig:primal_FFG}
\end{minipage}
\hspace{1cm}
\begin{minipage}[b]{0.45\linewidth}
\centering 
\includegraphics[width=1.0\columnwidth, keepaspectratio]{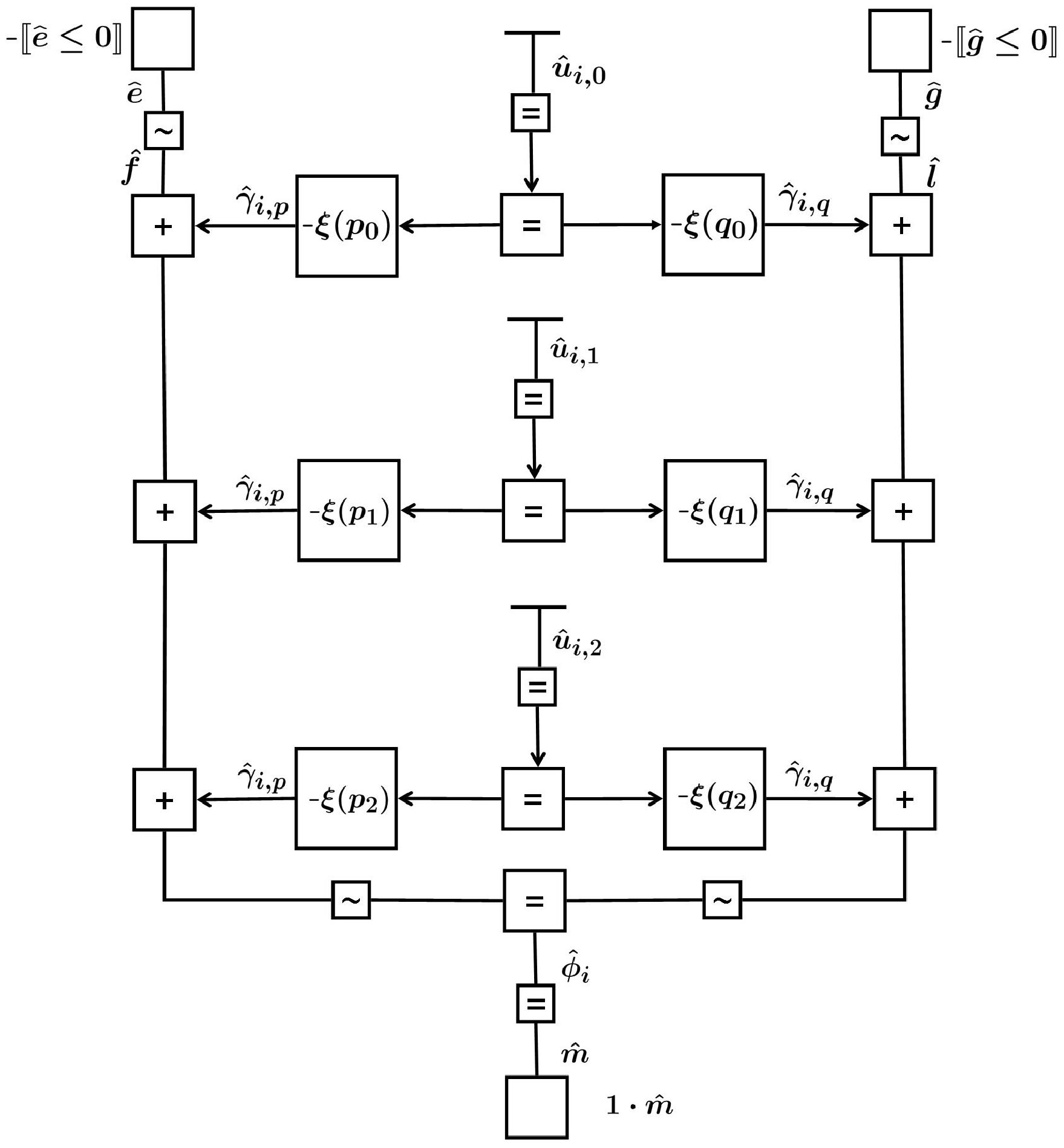}
\caption{FFG for the function $\hat{A}_{i}(\hat{\bs{u}}_{i})$. This FFG is dual to that of \reffig{fig:primal_FFG}. Here, for any variable $x$, $\hat{x}$ denotes the dual variable.} 
\label{fig:dual_FFG}
\end{minipage}
\end{figure*}
The dual function $\hat{A}_{i}(\hat{\bs{u}}_{i})$ is derived from the FFG of \reffig{fig:dual_FFG} as follows,
\begin{equation}
\hat{A}_{i}(\hat{\bs{u}}_{i}) = \hat{m} - \llbracket \hat{e} \leq  0 \rrbracket - \llbracket \hat{g} \leq 0 \rrbracket \label{eq:dual_1}
\end{equation}
where $\hat{m} = \hat{\phi}_{i}$ and
% 
% 
% \begin{equation}
% \hat{e} = -\hat{f} = \hat{\phi}_{i} - \langle \; -\hat{\bs{u}}_{i}, \; \Xi(\bs{p}) \; \rangle \Rightarrow -\left\llbracket \hat{e} \leq 0 \right\rrbracket = -\left\llbracket \hat{\phi}_{i} \leq \langle \; -\hat{\bs{u}}_{i}, \; \Xi(\bs{p}) \; \rangle \right\rrbracket \; .
% \label{eq:dual_2}
% \end{equation}
% 
% 
\begin{align} 
% \text{where, } & -\hat{m} = \hat{\phi}_{i}, \mbox{ and } \nonumber \\
& \hat{e} = -\hat{f} = -\hat{\phi}_{i} + \langle \; -\hat{\bs{u}}_{i}, \; \Xi(\bs{p}) \; \rangle \nonumber \\
\Rightarrow & -\left\llbracket \hat{e} \geq 0 \right\rrbracket = -\left\llbracket \hat{\phi}_{i} \le \langle \; -\hat{\bs{u}}_{i}, \; \Xi(\bs{p}) \; \rangle \right\rrbracket \label{eq:dual_2}
\end{align}
Similarly
\begin{eqnarray}
-\llbracket \hat{g} \geq 0 \rrbracket = -\left \llbracket \hat{\phi}_{i} \le \langle \; -\hat{\bs{u}}_{i}, \; \Xi(\bs{q}) \; \rangle \right\rrbracket \label{eq:dual_3}
\end{eqnarray}
From \refeq{eq:dual_1}, \refeqwo{eq:dual_2}, \refeqwo{eq:dual_3}
\begin{align*}
& \hat{A}_{i}(\hat{\bs{u}}_{i}) = \hat{\phi}_{i} - \llbracket \hat{\phi}_{i} \le \langle \; -\hat{\bs{u}}_{i}, \; \Xi(\bs{p}) \; \rangle \rrbracket - \llbracket \hat{\phi}_{i} \le \langle \; -\hat{\bs{u}}_{i}, \; \Xi(\bs{q}) \; \rangle \rrbracket  \nonumber \\
& \Rightarrow \hat{A}_{i}(\hat{\bs{u}}_{i}) = \hat{\phi}_{i} - \left \llbracket  \hat{\phi}_{i} \le \min_{\bs{a} \in \mathcal{A}_i} \langle \; -\hat{\bs{u}}_{i}, \; \Xi(\bs{a}) \; \rangle \right\rrbracket 
\end{align*}
% 
% 
% Same procedure can be used to derive the dual of $B_j(\bs{v}_j)$
% \begin{eqnarray*}
% \hat{B}_{j}(\hat{\bs{v}}_{j}) = \hat{\theta}_{j} - \left \llbracket  \hat{\theta}_{j} \le \min_{\bs{b} \in \mathcal{C}_j} \langle \; -\hat{\bs{v}}_{j}, \; \Xi(\bs{b}) \; \rangle \right\rrbracket 
% \end{eqnarray*}
% 

The same procedure can be used to derive the dual of $B_j(\bs{v}_j)$ as
\begin{eqnarray*}
\hat{B}_{j}(\hat{\bs{v}}_{j}) = \hat{\theta}_{j} - \left \llbracket  \hat{\theta}_{j} \le \min_{\bs{b} \in \mathcal{C}_j} \langle \; -\hat{\bs{v}}_{j}, \; \Xi(\bs{b}) \; \rangle \right\rrbracket \; .
\end{eqnarray*}
We use $\hat{A}_{i}(\hat{\bs{u}}_{i})$ and $\hat{B}_{j}(\hat{\bs{v}}_{j})$ to derive the dual of the LP represented by \refeq{eq:primal}, which is in the form of the maximization of the following augmented cost function,
\begin{align}
\sum_{i \in \mathcal{I}} &\hat{A}_{i}(\hat{\bs{u}}_{i}) + \sum_{j \in \mathcal{J}} \hat{B}_{j}(\hat{\bs{v}}_{j}) - \sum_{i \in \mathcal{I}} \llbracket \hat{\bs{f}}_{i} = -\hat{\bs{u}}_{i,0} \rrbracket \nonumber \\
&- \sum_{(i,j) \in \mathcal{E}} \llbracket \hat{\bs{u}}_{i,j} = -\hat{\bs{v}}_{j,i} \rrbracket - \sum_{i \in \mathcal{I}} \left\llbracket \hat{\bs{f}}_{i} = - \bs{\lambda}_{i}\right\rrbracket \label{eq:dual} 
\end{align}
The augmented cost function of \refeq{eq:dual} for the $(5,2)$ binary code is represented by the FFG of \reffig{fig:dual_LP_FFG}. 

The dual of PNBLPD can now be obtained from \refeq{eq:dual},

\medskip
\textbf{DNBLPD}:
\begin{align}
% \text{DNBLPD:} \quad \quad &  & &\nonumber \\
\text{max.} &\quad   \sum_{i \in \mathcal{I}} \hat{\phi}_{i} + \sum_{j \in \mathcal{J}} \hat{\theta}_{j} & &\nonumber  \\
\mbox{Subj. to } & & &\nonumber \\
\hat{\phi}_{i} &\leq \min_{\bs{a} \in \mathcal{A}_i}\left\langle-\hat{\bs{u}}_i, \Xi(\bs{a})\right\rangle &(&i \in \mathcal{I}), \nonumber \\
\hat{\theta}_{j} &\leq \min_{\bs{b} \in \mathcal{C}_j}\left\langle-\bs{\hat{v}}_j, \Xi(\bs{b})\right\rangle &(&j \in \mathcal{J}), \nonumber \\
\hat{\bs{u}}_{i,j} &= - \hat{\bs{v}}_{j,i} &(&(i,j) \in \mathcal{E}), \nonumber \\
\hat{\bs{u}}_{i,0} &= - \hat{\bs{f}}_i &(&i \in \mathcal{I}), \nonumber \\
\hat{\bs{f}}_i &= \bs{\lambda}_i &(&i \in \mathcal{I}). \nonumber 
\end{align}
\subsection{Softened Dual Linear Program}
We make use of the soft-minimum operator introduced in \cite{VoKo_06} and derive the ``Softened Dual Linear Program''. For any $\kappa \in \mathbb{R}$, $\kappa > 0$, the soft-minimum operator is defined as 
\begin{equation*}
\min_{l}{}^{(\kappa)} \{ z_l \} \triangleq -\frac{1}{\kappa} \log \left( \sum_{l} e^{-\kappa z_l}\right) = -\sym^{-1}\left(\sum_{l}\sym\Big({-z_l}\Big)\right)
\end{equation*}
where $\min_{l}{}^{(\kappa)} \{ z_l \} \le \min_{l} \{z_{l}\}$ with equality attained in the limit as $\kappa \to \infty$.  With this we define the softened dual linear program SDNBLPD which is the same as the DNBLPD except that $\min$ is replaced by $\min {}^{(\kappa)}$.
% 
% 
% \begin{eqnarray}
% \text{SDNBLPD:} && \nonumber \\
% \text{max.}	&& \sum_{i \in \mathcal{I}} \hat{\phi}_{i} + \sum_{j \in \mathcal{J}} \hat{\theta}_{j} \nonumber \\
% \mbox{Subject to } && \nonumber \\
% \hat{\phi}_{i} &\leq& \min_{\bs{a} \in \mathcal{A}_i}{}^{(\kappa)} \left\langle-\hat{\bs{u}}_i, \Xi(\bs{a})\right\rangle \quad (i \in \mathcal{I}), \nonumber \\
% \hat{\theta}_{j} &\leq& \min_{\bs{b} \in \mathcal{C}_j}{}^{(\kappa)} \left\langle-\bs{\hat{v}}_j, \Xi(\bs{b})\right\rangle \quad  (j \in \mathcal{J}), \nonumber \\
% \hat{\bs{u}}_{i,j} &=& - \hat{\bs{v}}_{j,i}\quad \quad \quad \quad \quad \quad((i,j) \in \mathcal{E}), \nonumber \\
% \hat{\bs{u}}_{i,0} &=& - \hat{\bs{f}}_i \quad \quad \quad \quad \quad \quad \quad (i \in \mathcal{I}), \nonumber \\
% \hat{\bs{f}}_i &=& \bs{\lambda}_i \quad \quad \quad \quad \quad \quad \quad (i \in \mathcal{I}). \nonumber 
% \end{eqnarray}
% 
% 
\section{Local Function}
In SDNBLPD, $\hat{\phi}_i$ and $\hat{\theta}_j$ are involved in only one inequality and hence we can replace these inequalities with equality without changing the optimal solution (the same is true of DNBLPD). With this, let us select an edge $(i,j) \in \mathcal{E}$ and assume that the variables associated to the rest of the edges are kept constant; then the ``local function" related to edge $(i,j)$ is
\begin{eqnarray} \label{eq:local_fun}
{h} \left(\hat{\bs{u}}_{i,j}\right) = \min_{\bs{a} \in \mathcal{A}_i} {}^{(\kappa)} \left\langle -\hat{\bs{u}}_i, \Xi(\bs{a})\right\rangle + \min_{\bs{b} \in \mathcal{C}_j} {}^{(\kappa)} \left\langle-\bs{\hat{v}}_j, \Xi(\bs{b})\right\rangle \; 
\end{eqnarray}
Though the soft-minimum operator is an approximation of the minimum, its advantage can be observed from \refeq{eq:local_fun}. Here, the local function would be non-differentiable without use of the soft-minimum operator and as we will see in the next section, the convexity and differentiability of ${h} \left(\hat{\bs{u}}_{i,j}\right)$ make it easier to treat mathematically.
\section{Low Complexity LP Decoding Algorithm for Nonbinary Linear Codes}
If the current values of variables $\hat{\bs{u}}_{i,j}, \hat{\phi}_{i}, \hat{\theta}_{j}$ related to edge $(i,j) \in \mathcal{E}$ are replaced with the new values (at the same time keeping variables related to other edges constant) such that ${h} \left(\hat{\bs{u}}_{i,j}\right)$ is maximized, then we can guarantee that the dual function also increases or else remains constant at its current value. The new value $\bar{u}^{(\alpha)}_{i,j}$ for each $\hat{u}^{(\alpha)}_{i,j}, \alpha \in \Re^{-}$ which maximizes ${h} \left(\hat{\bs{u}}_{i,j}\right)$ is given by
% which increases (or at least keep it constant) the dual function can be calculated by having,
\begin{align} \label{eq:new_u_alpha}
\bar{u}^{(\alpha)}_{i,j} \triangleq \argmax_{\hat{u}^{(\alpha)}_{i,j}} \;{h} \left(\hat{\bs{u}}_{i,j}\right) \quad \left(\forall \alpha \in \Re^{-}\right)
\end{align}
Once we have calculated $\bar{\bs{u}}_{i,j}$, we can update the variables $\hat{\phi}_i$ and $\hat{\theta}_j$ accordingly. The calculation of $\bar{\bs{u}}_{i,j}$ is given in the following lemma.
% Following lemma proves the how $\bar{\bs{u}}_{i,j}$ can be calculated.
%The calculation required for updating $\bar{\bs{u}}_{i,j}$ is given by \reflem{le:lemma_1},
\begin{lemma} \label{le:lemma_1}
The value of $\bar{u}_{i,j}^{(\alpha)}$ of \refeq{eq:new_u_alpha} can be calculated by
\begin{equation*}
\quad \bar{u}_{i,j}^{(\alpha)} = \frac{1}{2} \left( (V_{i,\bar{\alpha}} - V_{i, \alpha}) - (C_{j,\bar{\alpha}} - C_{j,\alpha}) \right)
\end{equation*}
where,
\begin{align*}
&V_{i,\bar{\alpha}} \triangleq - \min_{\underset{a_j \ne \alpha}{\bs{a} \in \mathcal{A}_i}} {}^{(\kappa)} \left\langle -\hat{\bs{u}}_i, \Xi({\bs{a}}) \right\rangle, \\
&V_{i,\alpha} \triangleq - \min_{\underset{a_j = \alpha}{\bs{a} \in \mathcal{A}_i}} {}^{(\kappa)} \left\langle -\tilde{\bs{u}}_i, \Xi(\tilde{\bs{a}}) \right\rangle, \\
&C_{j,\bar{\alpha}} \triangleq - \min_{\underset{b_i \ne \alpha}{\bs{b} \in \mathcal{C}_j}} {}^{(\kappa)} \left\langle -\hat{\bs{v}}_j, \Xi(\bs{b}) \right\rangle, \\
&C_{j,\alpha} \triangleq - \min_{\underset{b_i = \alpha}{\bs{b} \in \mathcal{C}_j}} {}^{(\kappa)} \langle -\tilde{\bs{v}}_j, \Xi(\tilde{\bs{b}}) \rangle. 
% C_{j,\alpha} \triangleq - \min_{\underset{b_i = \alpha}{\bs{b} \in \mathcal{C}_j}} {}^{(\kappa)} \langle -\tilde{\bs{v}}_j, \Xi(\tilde{\bs{b}}) \rangle &&
\end{align*}
Here the vectors $\tilde{\bs{u}}_{i}$ and $\tilde{\bs{a}}$ are the vectors $\hat{\bs{u}}_i$ and ${\bs{a}}$ respectively where the $j$-th position is excluded. Similarly, vectors $\tilde{\bs{v}}_{j}$ and $\tilde{\bs{b}}$ are obtained by excluding the $i$-th position from $\hat{\bs{v}}_j$ and ${\bs{b}}$ respectively.\\ \end{lemma}
\begin{proof}
\begin{align*}
{h} \left(\hat{\bs{u}}_{i,j}\right) &= \min_{\bs{a} \in \mathcal{A}_i} {}^{(\kappa)} \left\langle -\hat{\bs{u}}_{i},\Xi(\bs{a})\right\rangle  + \min_{\bs{b} \in \mathcal{C}_j} {}^{(\kappa)} \left\langle-\bs{\hat{v}}_j,\Xi(\bs{b})\right\rangle\\
% = -\frac{1}{\kappa} &\log \left( \sum_{\bs{a} \in \mathcal{A}_i} e^{+\kappa \left\langle \hat{\bs{u}}_{i},\; \Xi(\bs{a}) \right\rangle} \right) \\
% &- \frac{1}{\kappa} \log \left( \sum_{\bs{b} \in \mathcal{C}_j} e^{+\kappa \left\langle \hat{\bs{v}}_{j}, \;\Xi(\bs{b}) \right\rangle} \right) \\
= -&\sym^{-1} \left( \sum_{\bs{a} \in \mathcal{A}_i} \sym \Big({\left\langle \hat{\bs{u}}_{i},\; \Xi(\bs{a}) \right\rangle} \Big) \right) \\
&- \sym^{-1} \left( \sum_{\bs{b} \in \mathcal{C}_j} \sym \Big({\left\langle \hat{\bs{v}}_{j}, \;\Xi(\bs{b}) \right\rangle}\Big) \right) \\
% % 
% % 
% = -\frac{1}{\kappa} &\log \left( \sum_{\bs{a} \in \mathcal{A}_i} e^{+\kappa \left\langle \hat{\bs{u}}_{i,j}, \xi(a_j) \right\rangle + \kappa \left\langle \tilde{\bs{u}}_{i},\; \Xi(\tilde{\bs{a}}) \right\rangle} \right) \\
% &- \frac{1}{\kappa} \log \left( \sum_{\bs{B} \in \mathcal{C}_j} e^{-\kappa \left\langle \hat{\bs{u}}_{i,j}, \xi(b_i) \right\rangle + \kappa \left\langle \tilde{\bs{v}}_{j},\; \Xi(\tilde{\bs{b}}) \right\rangle} \right) \\
% 
% 
% Following contains vector split first and then sum split
% = -&\sym^{-1} \left( \sum_{\bs{a} \in \mathcal{A}_i} \sym \Big( {\left\langle \hat{\bs{u}}_{i,j}, \xi(a_j) \right\rangle + \left\langle \tilde{\bs{u}}_{i},\; \Xi(\tilde{\bs{a}}) \right\rangle}\Big) \right) \\
% &- \sym^{-1} \left( \sum_{\bs{b} \in \mathcal{C}_j} \sym \Big({-\left\langle \hat{\bs{u}}_{i,j}, \xi(b_i) \right\rangle + \langle \tilde{\bs{v}}_{j},\; \Xi(\tilde{\bs{b}}) \rangle}\Big) \right) \\
% 
% 
% 
% 
=-&\sym^{-1}\left( \sum_{\underset{a_j \ne \alpha} {\bs{a} \in \mathcal{A}_i}} \sym \Big({\left\langle \hat{\bs{u}}_{i},\; \Xi(\hat{\bs{a}}) \right\rangle}\Big)  %\right. \\
%&+ \left. 
+\sum_{\underset{a_j = \alpha} {\bs{a} \in \mathcal{A}_i}} \sym \Big(\langle \hat{\bs{u}}_{i},\; \Xi(\hat{\bs{a}}) \rangle\Big) \right)\\
-&\sym^{-1} \left( \sum_{\underset{b_i \ne \alpha}{\bs{b} \in \mathcal{C}_j}} \sym \Big(\langle \hat{\bs{v}}_{j},\; \Xi(\hat{\bs{b}}) \rangle\Big)  
%\right. \\
%&\left.\quad\quad\quad
+ \sum_{\underset{b_i = \alpha}{\bs{b} \in \mathcal{C}_j}} \sym \Big(\langle \hat{\bs{v}}_{j},\; \Xi(\hat{\bs{b}}) \rangle\Big) \right) \\
% 
% 
% %
% %
% \end{align*}
% \begin{align*}
% {h} \left(\hat{\bs{u}}_{i,j}\right) 
% = -\frac{1}{\kappa} &\log \left( \sum_{\underset{a_j \ne \alpha} {\bs{a} \in \mathcal{A}_i}} e^{+\kappa \left\langle \hat{\bs{u}}_{i,j}, \xi(a_j) \right\rangle  + \kappa \left\langle \tilde{\bs{u}}_{i},\; \Xi(\tilde{\bs{a}}) \right\rangle}  \right. \\
% &\quad\quad\quad+ \left. \sum_{\underset{a_j = \alpha} {\bs{a} \in \mathcal{A}_i}} e^{+\kappa \hat{u}_{i,j}^{(\alpha)} + \kappa \left\langle \tilde{\bs{u}}_{i},\; \Xi(\tilde{\bs{a}}) \right\rangle} \right)\\
=  -&\sym^{-1}\left( \sum_{\underset{a_j \ne \alpha} {\bs{a} \in \mathcal{A}_i}} \sym \Big({\left\langle \hat{\bs{u}}_{i},\; \Xi(\hat{\bs{a}}) \right\rangle}\Big)  \right. \\
&\quad\quad +\left. \sum_{\underset{a_j = \alpha} {\bs{a} \in \mathcal{A}_i}} \sym \Big({\hat{u}_{i,j}^{(\alpha)} + \langle \tilde{\bs{u}}_{i},\; \Xi(\tilde{\bs{a}}) \rangle}\Big) \right)\\
% %
% %
% \quad \quad -\frac{1}{\kappa} &\log \left( \sum_{\underset{b_i \ne \alpha}{\bs{b} \in \mathcal{C}_j}} e^{- \kappa \left\langle \hat{\bs{u}}_{i,j}, \xi(b_i) \right\rangle  + \kappa \left\langle \tilde{\bs{v}}_{j},\; \Xi(\tilde{\bs{b}}) \right\rangle}  \right. \\
% &\left.\quad\quad\quad+ \sum_{\underset{b_i = \alpha}{\bs{b} \in \mathcal{C}_j}} e^{-\kappa \hat{u}_{i,j}^{(\alpha)} + \kappa \left\langle \tilde{\bs{v}}_{j},\; \Xi(\tilde{\bs{b}}) \right\rangle} \right) \\
-&\sym^{-1} \left( \sum_{\underset{b_i \ne \alpha}{\bs{b} \in \mathcal{C}_j}} \sym \Big({\langle \hat{\bs{v}}_{j},\; \Xi(\hat{\bs{b}}) \rangle}\Big)  \right. \\
&\left.\quad\quad+ \sum_{\underset{b_i = \alpha}{\bs{b} \in \mathcal{C}_j}} \sym \Big({-\hat{u}_{i,j}^{(\alpha)} + \langle \tilde{\bs{v}}_{j},\; \Xi(\tilde{\bs{b}}) \rangle}\Big) \right) \\
% %
% %
% = -\frac{1}{\kappa} &\log \left( e^{+\kappa V_{i,\bar{\alpha}}}  + e^{+  \hat{u}_{i,j}^{(\alpha)}} e^{+\kappa V_{i,\alpha}} \right) \\
% &- \frac{1}{\kappa} \log \left( e^{+\kappa C_{j,\bar{\alpha}}} + e^{-\kappa \hat{u}_{i,j}^{(\alpha)}} e^{+\kappa C_{j,\alpha}} \right)
= -&\sym^{-1} \bigg( \sym \Big({V_{i,\bar{\alpha}}}\Big)  + \sym \Big({\hat{u}_{i,j}^{(\alpha)}}\Big) \sym \Big({V_{i,\alpha}}\Big) \bigg) \\
&- \sym^{-1} \bigg( \sym \Big({C_{j,\bar{\alpha}}}\Big) + \sym \Big({- \hat{u}_{i,j}^{(\alpha)}}\Big) \sym \Big({C_{j,\alpha}}\Big) \bigg)
\end{align*}
Now to maximize ${h} \left(\hat{\bs{u}}_{i,j}\right)$, we set
\begin{align*}
\frac{\partial h\left( \hat{\bs{u}}_{i,j} \right)}{\partial \hat{u}_{i,j}^{(\alpha)}} =
% \overset{!}{=} 0 \\
% \Rightarrow 
% - \frac{1}{\kappa} &\frac{+\kappa e^{+\kappa \hat{u}_{i,j}^{(\alpha)}} e^{+\kappa V_{i,\alpha}}} {e^{+\kappa V_{i,\bar{\alpha}}}  + e^{+\kappa \hat{u}_{i,j}^{(\alpha)}} e^{+\kappa V_{i,\alpha}}} &\\
% -\frac{1}{\kappa} &\frac{-\kappa e^{-\kappa \hat{u}_{i,j}^{(\alpha)}} e^{+\kappa C_{j,\alpha}}} {e^{+\kappa C_{j,\bar{\alpha}}} + e^{-\kappa \hat{u}_{i,j}^{(\alpha)}} e^{+\kappa C_{j,\alpha}}} = 0,&
% 
% 
-\frac{1}{\kappa} &\frac{+\kappa \; \sym \Big({\hat{u}_{i,j}^{(\alpha)}}\Big) \sym \Big({V_{i,\alpha}}\Big)} {\sym \Big({V_{i,\bar{\alpha}}}\Big)  + \sym \Big({\hat{u}_{i,j}^{(\alpha)}}\Big) \; \sym \Big({V_{i,\alpha}}\Big)} &\\
-\frac{1}{\kappa} \;\; &\frac{-\kappa \; \sym \Big({-{u}_{i,j}^{(\alpha)}}\Big) \sym \Big({C_{j,\alpha}} \Big)} {\sym \Big({C_{j,\bar{\alpha}}}\Big) + \sym \Big({-\hat{u}_{i,j}^{(\alpha)}}\Big) \; \sym \Big({C_{j,\alpha}}\Big)} = 0&
\end{align*}
\begin{align*}
% &\Rightarrow &e^{+\kappa \hat{u}_{i,j}^{(\alpha)}} &e^{+\kappa \left(V_{i,\alpha} + C_{j,\bar{\alpha}} \right)} +  e^{+\kappa \left(V_{i,\alpha} + C_{j,{\alpha}} \right)} \\
% % % %
% % % %
% & & &= e^{-\kappa \hat{u}_{i,j}^{(\alpha)} } e^{+\kappa \left(V_{i,\bar{\alpha}} + C_{j,\alpha} \right)} +  e^{+\kappa \left(V_{i,\alpha} + C_{j,{\alpha}} \right)} \\
% 
&\Rightarrow \sym \Big({\hat{u}_{i,j}^{(\alpha)}}\Big) \; \sym \Big({V_{i,\alpha} + C_{j,\bar{\alpha}} }\Big) = \sym \Big({-\hat{u}_{i,j}^{(\alpha)} }\Big) \; \sym \Big(V_{i,\bar{\alpha}} + C_{j,\alpha}\Big) \\
% %
% %
&\Rightarrow \hat{u}_{i,j}^{(\alpha)} + (V_{i,\alpha} +  C_{j,\bar{\alpha}}) = -\hat{u}_{i,j}^{(\alpha)} + (V_{i,\bar{\alpha}} + C_{j,\alpha}) \; . \\
% \end{align*}
&\text{This yields,} \\
% \begin{align*}
&\quad\quad\quad\bar{u}_{i,j}^{(\alpha)} = \frac{1}{2} ((V_{i,\bar{\alpha}} - V_{i, \alpha}) - (C_{j,\bar{\alpha}} - C_{j,\alpha})) \; .
\end{align*}
\end{proof}
% \medskip
% The update equation of \reflem{le:lemma_1} looks very similar to the one suggested in \cite{VoKo_06}. \reflem{le:lemma_1} is in fact the generalization of the Lemma 3 of \cite{VoKo_06} for nonbinary codes and in case of code over $GF(2)$, these equations turns out to be the same as the one given in \cite{VoKo_06}. 
\reflem{le:lemma_1} is a generalization of Lemma 3 of \cite{VoKo_06} to the case of nonbinary codes. One visible difference between the binary case and the present generalization is in the calculation of $V_{i,\bar{\alpha}}$ and $C_{j,\bar{\alpha}}$. Here in the case of nonbinary codes, the calculation of $V_{i,\bar{\alpha}}$ does not exclude the $j$-th entry from $\bs{a} \in \mathcal{A}_i$ and $\hat{\bs{u}}_{i,j}$; similarly the calculation of $C_{j,\bar{\alpha}}$ does not exclude the $i$-th entry from $\bs{b} \in \mathcal{C}_j$ and $\hat{\bs{v}}_{j,i}$. Note that this is not inconsistent since $\hat{u}^{({\alpha})}_{i,j}$ is never used to update itself. 
% but it has a simple explanation. 
Here the calculation of $V_{i,\bar{\alpha}}$ and $C_{j,\bar{\alpha}}$ requires $\bar{\alpha} \in \Re \setminus \{0, \alpha\}$ and hence $\xi(\bar{\alpha})$ is always multiplied with the corresponding $\hat{u}^{(\bar{\alpha})}_{i,j}$. This ensures that $\hat{u}^{({\alpha})}_{i,j}$ is not used for calculating $\bar{u}^{(\alpha)}_{i,j}$.

As mentioned in \cite{VoKo_06}, the update equation given in Lemma 3 of \cite{VoKo_06}
% for the low-complexity LP decoding algorithm for the binary codes 
can be efficiently computed with the help of the variable and check node calculations of the (binary) SP algorithm. Due to this, the complexity of computing $(C_{j,\bar{\alpha}} - C_{j,{\alpha}})$ is $O(d)$ for binary codes.
% This direct relationship with SP algorithm reduces the complexity of computing $-(C_{j,\bar{\alpha}} - C_{j,{\alpha}})$ which is in $O(d)$. 
On the other hand, in case of nonbinary codes the mapping $\Xi$ used in NBLPD transforms the nonbinary linear codes $\mathcal{A}_i$ (repetition code) and $\mathcal{C}_j$ (SPC code) into nonlinear binary codes $\mathcal{A}^{NL}_i = \{\Xi(\bs{a}) : \forall \bs{a} \in \mathcal{A}_i\}$ and $\mathcal{C}^{NL}_j = \{\Xi(\bs{b}) : \forall \bs{b} \in \mathcal{C}_j\}$ respectively.
% and hence 
% there is no direct relationship between the update equations of \reflem{le:lemma_1} and those of the nonbinary Sum-Product algorithm. 
% Instead of being related to the nonbinary SP algorithm, 
Here, the computation of $(V_{i,\bar{\alpha}} - V_{i,{\alpha}})$ and $(C_{j,\bar{\alpha}} - C_{j,{\alpha}})$ is related to the SP decoding of nonlinear binary codes %$\mathcal{A}^{NL}_i = \{\Xi(\bs{a}) : \forall \bs{a} \in \mathcal{A}_i\}$ and $\mathcal{C}^{NL}_j = \{\Xi(\bs{b}) : \forall \bs{b} \in \mathcal{C}_j\}$.
$\mathcal{A}^{NL}_i$ and $\mathcal{C}^{NL}_j$. $\mathcal{A}_i$ and $\mathcal{C}_j$ are duals of each other, however such relationship 
between $\mathcal{A}^{NL}_i$ and $\mathcal{C}^{NL}_j$ 
%is not yet established.
% exsists between $\mathcal{A}^{NL}_i$ and $\mathcal{C}^{NL}_j$ or not. 
% doesn't exist between $\mathcal{A}^{NL}_i$ and $\mathcal{C}^{NL}_j$.
% however relationship between $\mathcal{A}^{NL}_i$ and $\mathcal{C}^{NL}_j$ is not immediately clear. 
is not so simple.
Hence the computation of $(C_{j,\bar{\alpha}} - C_{j,{\alpha}})$ in the dual domain requires further investigation.
% Hence if the computation of $(C_{j,\bar{\alpha}} - C_{j,{\alpha}})$ can be carried out in dual domain or not is still an open question.
% Hence if the computation of $(C_{j,\bar{\alpha}} - C_{j,{\alpha}})$ can be carried out in dual domain or not is not immediately clear and requires further investigation.
%  and requires further investigation.
% 
% 
% Consequently, the algorithm must run through all possible codewords of the SPC code $\mathcal{C}_j$ which means that the complexity of computing $(C_{j,\bar{\alpha}} - C_{j,{\alpha}})$ is $O(q^{(d-1)})$ for nonbinary codes.

One option to compute $(C_{j,\bar{\alpha}} - C_{j,{\alpha}})$ 
% is through the exhaustive search i.e. to go through all possible codewords of the SPC code $\mathcal{C}_j$. 
is by going through all possible codewords of the SPC code $\mathcal{C}_j$ exhaustively.
In this case the complexity of computing $(C_{j,\bar{\alpha}} - C_{j,{\alpha}})$ is $O(q^{(d-1)})$. 
However it is also possible to rewrite the equations for $C_{j,\bar{\alpha}} \; \text{ and } \; C_{j,\alpha}$ as follows, %used in \reflem{le:lemma_1} 
\begin{align*}
% &C_{j,\bar{\alpha}} = \sym^{-1} \sum_{\underset{b_i \ne \alpha}{\bs{b} \in \mathcal{C}_j}} \sym \Big(\left\langle \hat{\bs{v}}_j, \Xi(\bs{b}) \right\rangle \Big)\\  
% & \Rightarrow 
\sym \Big(C_{j,\bar{\alpha}}\Big) = \sum_{\underset{b_i \ne \alpha}{\bs{b} \in \mathcal{C}_j}} \sym \Big(\left\langle \hat{\bs{v}}_j, \Xi(\bs{b}) \right\rangle \Big)
\end{align*}
Similarly
\begin{align*}
& \sym \Big(C_{j,\alpha}\Big) = \sum_{\underset{b_i = \alpha}{\bs{b} \in \mathcal{C}_j}} \sym \Big(\langle \tilde{\bs{v}}_j, \Xi(\tilde{\bs{b}}) \rangle\Big)
\end{align*}
It can be observed from the above equations that the calculation of the $C_{j,\bar{\alpha}} \; \text{and} \; C_{j,\alpha}$ is in the form of the marginalization of a product of functions.
% sum of the products of the probabilities representated by $\sym({\hat{v}_{j,i}}), i \in \mathcal{I}_j$. 
Hence it is possible to compute $C_{j,\bar{\alpha}} \text{ and } C_{j,\alpha}$ with the help of a trellis based variant of the SP algorithm. 
% of the SPC code $\mathcal{C}_j$. 
The complexity of computing $(C_{j,\bar{\alpha}} - C_{j,{\alpha}})$ with the help of the trellis of the SPC code $\mathcal{C}_j$ is linear in the maximum check-node degree $d$. However, this trellis based approach is still under investigation and is not used for the simulation result given in the \refsec{sec_result}.
% With this observation we can  and these terms can also be computed with the help of the trellis of the SPC code $\mathcal{C}_j$.
% of input probabilities (i.e. $e^{-\langle \bs{v}_j, \Xi(\bs{b}) \rangle}$).

We can now formulate the decoding algorithm with the help of the update equation given in \reflem{le:lemma_1}. We select an edge $(i,j) \in \mathcal{E}$ and calculate $\bar{\bs{u}}_{i,j}$ from \reflem{le:lemma_1}. Then $\hat{\phi}_i, \; \hat{\theta}_j$ and the objective function are updated accordingly. One \textit{iteration} is completed when all edges $(i,j) \in \mathcal{E}$ are updated cyclically. This is a coordinate-ascent type algorithm and its convergence may be proved in the same manner as in Lemma 4 of \cite{VoKo_06}.
\begin{lemma} \label{le:convergence}
We assume $d \geq 3$ for a given parity-check matrix $\mathcal{H}$ of the code $\mathcal{C}$. If we update all edges $(i,j) \in \mathcal{E}$ cyclically with the update equation given in \reflem{le:lemma_1}, then the objective function of SDNBLPD converges to its maximum.
% Lets assume that $d \geq 3$ for a given parity-check matrix $\mathcal{H}$ of the code $\mathcal{C}$. If we update all edges $(i,j) \in \mathcal{E}$ with the help of \reflem{le:lemma_1} according to the circular schedule, then the objective function of SDNBLPD is maximized.
\end{lemma}
\begin{proof}
% (should we write See Lemma 4 of \cite{VoKo_06}.??) 
The proof is essentially the same as that of Lemma 4 of \cite{VoKo_06}.
\end{proof}
\medskip
The algorithm terminates after a fixed number of iterations or when it finds a codeword. Knowing the solution of SDNBLPD does not give an estimate of the codeword directly. However, an estimate of the $i$-th symbol $\hat{\bs{c}}_i$ can be obtained from the vector $\hat{\bs{u}}_{i,j}$. For this we define, 
\begin{equation*}
\hat{{x}}^{(\alpha)}_i = \sum_{j \in \mathcal{J}_i} -\hat{u}_{i,j}^{(\alpha)} 
\end{equation*}
It is possible that the value of $\hat{{x}}^{(\alpha)}_i$ is zero. In this case, the corresponding symbol is erased. Otherwise the symbol estimate is obtained as follows:
\begin{eqnarray}
% \text{estimate of $i$-th symbol} \quad 
\hat{\bs{c}}_i &=& \xi^{-1}(\hat{\bs{f}}_i) \label{eq:decision}
\end{eqnarray}
where
\begin{eqnarray*}
% \text{ where } && \nonumber \\ %\hat{\bs{f}}_i &=& (\hat{{f}}^{(\alpha)}_i)_{\alpha \in \Re^{-}} \text{ and, } \\
\hat{{f}}^{(\alpha)}_i &=& \left \{ 
\begin{array}{l l}
 \text{1}, & \text{ if } \hat{{x}}^{(\alpha)}_i < \text{0}, \nonumber \\
 \text{0}, & \text{ if } \hat{{x}}^{(\alpha)}_i > \text{0} \; .  \\
 \end{array} \right . 
\end{eqnarray*}
If more than one $\hat{{f}}^{(\alpha)}_i$ is assigned the value $1$, then the inverse function $\xi^{-1}$ cannot be invoked to give the estimate of $\bs{c}_i$. However, the constraints of NBLPD also enforce $\sum_{\alpha \in \Re^{-}} {{f}}^{(\alpha)}_i = 1$, and hence we will never get such a configuration of $\hat{{f}}^{(\alpha)}_i$.

The advantage of using the soft-minimum operator is evident from \reflem{le:lemma_1}. However, for practical implementation we are interested in $\kappa \to \infty$. As mentioned earlier, in limit of $\kappa \to \infty$, the soft-minimum operator is same as the minimum which requires less computation. The following lemma considers $\kappa \to \infty$.
% 
% The complexity of soft-minimum operator is prohibitively large for practical implementation. However, if we take $k \to \infty$ then soft-min operation reduces to the standard minimum operation. 
% The advantage of using soft-minimum operator is evident from the \reflem{le:lemma_1}. However, the complexity of calculating soft-min is prohibitively large for practical implementation. 
%Following Lemma considers the case of $k \to \infty$ during the maximization of function $h(\hat{\bs{u}}_{i,j})$.
\begin{lemma} \label{le:lemma_k_infty}
In the limit of $\kappa \to \infty$, the function $h(\hat{\bs{u}}_{i,j})$ is maximized by any value $\hat{{u}}_{i,j}^{(\alpha)}$ that lies in the closed interval between
% the value of $\bar{u}_{i,j}^{(\alpha)}$ of \refeq{eq:new_u_alpha} lies in the interval between
% if $\hat{u}_{i,j}^{(\alpha)}, \alpha \in \Re^{-}$ is replaced by any value that lies in the closed interval between 
\begin{eqnarray*}
(V_{i,\bar{\alpha}} -V_{i,{\alpha}})  &\text{ and }& -(C_{j,\bar{\alpha}} -C_{j,{\alpha}})
\end{eqnarray*}
% $(V_{i,\bar{\alpha}} -V_{i,{\alpha}})  \text{ and } -(C_{j,\bar{\alpha}} -C_{j,{\alpha}})$ 
% while $\hat{u}_{i,j}^{(\bar{\alpha})}, \forall \bar{\alpha} \in \Re \setminus \{0,\alpha\}$ are held constant, the function $h(\hat{\bs{u}}_{i,j})$ is maximized. 
% the function $h(\hat{\bs{u}}_{i,j})$ is maximized by any value that lies in the closed interval between
% \begin{eqnarray*}
% (V_{i,\bar{\alpha}} -V_{i,{\alpha}})  &\text{ and }& -(C_{j,\bar{\alpha}} -C_{j,{\alpha}})
% \end{eqnarray*}
where
\begin{eqnarray*}
V_{i,\bar{\alpha}} \triangleq - \min_{\underset{a_j \ne \alpha}{\bs{a} \in \mathcal{A}_i}} \left\langle -\hat{\bs{u}}_i, \Xi({\bs{a}}) \right\rangle && C_{j,\bar{\alpha}} \triangleq - \min_{\underset{b_i \ne \alpha}{\bs{b} \in \mathcal{C}_j}} \left\langle -\hat{\bs{v}}_j, \Xi(\bs{b}) \right\rangle, \\
V_{i,\alpha} \triangleq - \min_{\underset{a_j = \alpha}{\bs{a} \in \mathcal{A}_i}} \left\langle -\tilde{\bs{u}}_i, \Xi(\tilde{\bs{a}}) \right\rangle && C_{j,\alpha} \triangleq - \min_{\underset{b_i = \alpha}{\bs{b} \in \mathcal{C}_j}} \langle -\tilde{\bs{v}}_j, \Xi(\tilde{\bs{b}}) \rangle.
\end{eqnarray*}
\end{lemma}
\begin{proof}
The proof of the lemma is same as that of Lemma 5 of \cite{VoKo_06}.
\end{proof}
\medskip
Now we can update edges $(i,j) \in \mathcal{E}$ cyclically where the $\bar{\bs{u}}_{i,j}$ is calculated according to \reflem{le:lemma_k_infty}. However, in this case, we cannot guarantee convergence of the algorithm. This is because for $\kappa \to \infty$ the objective function is not everywhere differentiable and it is not possible to use the same argument as in \reflem{le:convergence}. This problem is also discussed in Conjecture 6 of \cite{VoKo_06}. 
%As in the case of binary low-complexity LP decoding \cite{VoKo_06}, we also think that the algorithm which uses update equation from \reflem{le:lemma_k_infty} can't get stuck at the suboptimal point. However, we don't have a proof for this claim so far.  
After the algorithm terminates, \refeq{eq:decision} can be used to get the estimate for each symbol.
\begin{figure}
\centering
\includegraphics[clip=true,width=1.0\columnwidth, keepaspectratio]{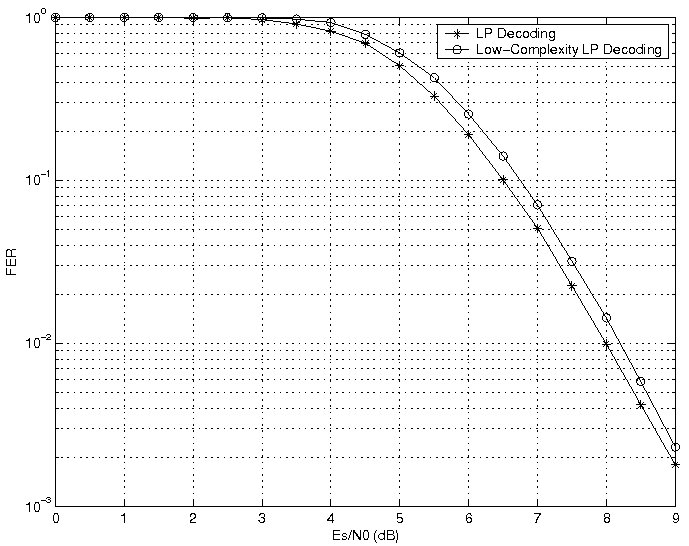}
\caption{Frame Error Rate for the example $[80, 48]$ quaternary LDPC Code under QPSK modulation. The performance of the low-complexity LP decoding algorithm is compared with that of solving NBLPD using the simplex algorithm.} 
\label{fig:FER}
\end{figure}
\section{Results} \label{sec_result}
In this section we present simulation results for the proposed algorithm. The update equation of \reflem{le:lemma_k_infty} is used for simulations. Calculation of the $(C_{j,\bar{\alpha}} -C_{j,{\alpha}})$ is carried out with exhaustive search over all codewords of SPC code $\mathcal{C}_j$. We use the LDPC code of length $n = 80$ over $\mathbb{Z}_4$. This code has rate $R(\mathcal{C}) = 0.6$ and constant check-node degree of $5$. Its parity check matrix can be constructed as follows:
\begin{eqnarray*}
\mathcal{H}_{j,i} = \left \{
\begin{array}{l l}
\text{1}, & \text{ if } i-j = \{\text{0, 41, 48}\} \\
\text{3}, & \text{ if } i-j = \{\text{8, 25}\} \\
\text{0}, & \text{ otherwise.}
\end{array} \right .
\end{eqnarray*}
We assume transmission over the AWGN channel where the nonbinary symbols are directly mapped to quaternary phase-shift keying (QPSK) signals. The same LDPC code was also used in \cite{FlSk_09}.

\reffig{fig:FER} shows the error correcting performance curve for above mentioned LDPC code. The curve marked ``LP Decoding'' uses the LP decoding algorithm of \cite{FlSk_09} with the simplex LP solver. All results are obtained by simulating up to $500$ frame errors per simulation point. The error correcting performance of low-complexity LP decoding algorithm is within $0.2$ dB of the LP decoder. It is important to note that the worst case time complexity of the simplex method has been shown to be exponential in the number of variables (i.e. block length). 
% Though the complexity of check-node calculation with exhaustive search is exponential in the check-node degree, 
In contrast, the complexity of the low-complexity LP decoding is linear in the block length.

\section{Conclusion and Future Work}
In this paper we introduced low-complexity LP decoding algorithm for nonbinary linear codes. Building on the work of Flanagan \textit{et al.} \cite{FlSk_09} and Vontobel \textit{et al.} \cite{VoKo_06}, we derived the update equations of \reflem{le:lemma_1} \& \reflem{le:lemma_k_infty}. The complexity of the proposed algorithm is linear in the block length and hence it can also be used for moderate and long block length codes. However, its complexity is dominated by the maximum check node degree and the number of elements in the nonbinary alphabet. The main problem is that of the check node calculations. %, which are not performed in the dual domain as in the case of binary low-complexity LP decoding.
The binary repetition code is the dual of the binary SPC code and this fact is utilized in binary low-complexity LP decoding algorithm to reduce the computational complexity.
% The nonbinary repetition code is the dual of the nonbinary SPC code and this fact is utilized in nonbinary SP algorithm to reduce the computational complexity. 
% However, for the nonbinary case, analogous relationships between the corresponding nonlinear binary codes remain to be established. We are currently investigating the relationship between these two nonlinear binary codes and also the trellis based variant of the sum-product algorithm with the goal of leading towards a complexity reduction in the check node calculations.	
However, for the nonbinary case, the relationship between the corresponding nonlinear binary codes is not so simple. We are currently investigating 
% the relationship between these two nonlinear binary codes and also the
dual domain methods for check node processing as well as 
variants of the sum-product algorithm which operate directly on the trellis of the nonbinary code, with the goal of leading towards a complexity reduction in the check node calculations.
% . Other approach for complexity reduction is to use the trellis of the nonbinary SPC code and apply the BCJR type algorithm for check node calculation.
% 
% The main problem being 
% calculations related to the nonlinear binary code generated by the check . 
% It would be interesting to further explore the nonlinear binary codes generated at the variable and check node to reduce 
%%%%%%%%%%%%%%%%%%%%%%%%%%%%%%%%%%%%%%%%%%%%%%%%%%%%%%%%%%%%%%%%%%%%%%%%%%%%%%%%
\section{ACKNOWLEDGMENTS}
The authors would like to thank P. O. Vontobel for many helpful suggestions and comments. This work was supported in part by Claude Shannon Institute for Discrete Mathematics, Coding and Cryptography, UCD, Ireland.

%%%%%%%%%%%%%%%%%%%%%%%%%%%%%%%%%%%%%%%%%%%%%%%%%%%%%%%%%%%%%%%%%%%%%%%%%%%%%%%%

\end{document}